\documentclass[11pt]{article}

\usepackage[all]{xy}           
\usepackage{amsmath,amsfonts,amssymb,amsthm,textcomp}
\usepackage{eucal}
\usepackage{epsfig}
\usepackage{lscape}

\def\v #1{\vert #1\vert}             

\def\m #1 #2{(-1)^{{\v #1} {\v #2}}} 

\theoremstyle{plain}
\newtheorem{theorem}{Theorem}

\newtheorem{proposition}[theorem]{Proposition}

\theoremstyle{definition}
\newtheorem{definition}[theorem]{Definition}

\def\<#1>{\langle#1\rangle}

\usepackage{epsfig}
\usepackage{amsmath}
\usepackage{epic}
\usepackage{amssymb}
\usepackage{fancybox}
\usepackage{wrapfig}
\usepackage[font=footnotesize,labelsep=period]{caption}
\usepackage{tikz-cd}
\usepackage{float}

\usepackage{booktabs}

\usepackage[all]{xy}

\usepackage{fancyhdr}

\usepackage{fancyhdr}

\usepackage{appendix}

\begin{document}

\centerline{\Large \bf Cosymplectic and contact structures to resolve}\vskip 0.25cm
\centerline{\Large \bf  time-dependent and dissipative hamiltonian systems}

\medskip
\medskip

\centerline{M. de Le\'on and C.
Sard\'on}
\vskip 0.5cm
\centerline{Instituto de Ciencias Matem\'aticas, Campus Cantoblanco}\vskip 0.2cm
\centerline{Consejo Superior de Investigaciones Cient\'ificas}
\vskip 0.2cm
\centerline{C/ Nicol\'as Cabrera, 13--15, 28049, Madrid. SPAIN}

\begin{abstract}

In this paper, we apply the geometric Hamilton--Jacobi theory to obtain solutions of classical hamiltonian systems 
that are either compatible with a cosymplectic or a contact structure. As it is well known, the first structure plays a central role
in the theory of time-dependent hamiltonians, whilst the second is here used to treat classical hamiltonians including dissipation terms.

The interest of a geometric Hamilton--Jacobi equation is the primordial observation that if a hamiltonian vector field $X_{H}$ can be projected into a
configuration manifold by means of a 1-form $dW$, then the integral curves of the projected
vector field $X_{H}^{dW}$ can be transformed into integral curves of $X_{H}$ provided that $W$ is a solution of the Hamilton--Jacobi equation.
In this way, we use the geometric Hamilton--Jacobi theory to derive solutions of physical systems with a time-dependent
hamiltonian formulation or including dissipative terms. Explicit, new expressions for a geometric
Hamilton--Jacobi equation are obtained on a cosymplectic and a contact manifold. These equations are later
used to solve physical examples containing explicit time dependence, as it is the case of a unidimensional trigonometric system,
and two dimensional nonlinear oscillators as Winternitz--Smorodinsky oscillators. For explicit dissipative behavior, we solve 
the example of a unidimensional damped oscillator.

\end{abstract}

\section{Introduction}
%

In this paper we are concerned with {almost cosymplectic structures} and their application in classical hamiltonian mechanics.
By an {\it almost cosymplectic structure} we understand a $2n+1$-dimensional manifold equipped with a one-form $\eta$ and a two-form $\Omega$ such that $\eta\wedge \Omega^n$ is a volume form.
In particular, we will study two cases of almost cosymplectic manifolds. On one hand, the case of {cosymplectic manifolds} \cite{Blair,Cape,LeonSara,LeonTuyn}, and on the other hand, the case of
{contact manifolds} \cite{Blair,boothwang,etnyre,Godbillon}.
Cosymplectic manifolds have shown their usefulness in theoretical physics, as in gauge theories of gravity, branes and string theory \cite{Becker,Cham,Hitchin}. 
Among the early studies of cosymplectic manifolds we mention A. Lichnerowicz \cite{Lich2,Lich}, who studied the Lie algebra of infinitesimal automorphisms of a 
cosymplectic manifold, in analogy with the symplectic case. 
Posteriously, some works have endowed cosymplectic manifolds with a Riemannian metric, the so-called {coK\"ahler manifolds} \cite{oku}. These
are the odd dimensional counterpart of K\"ahler manifolds.
Another important role of the cosymplectic theory is the reduction theory to reduce time-dependent hamiltonians by symmetry groups \cite{Albert,CLMMdD,CLM}.
But since very foundational papers by P. Libermann, very sporadic papers have appeared on cosymplectic settings. It is our future intention to provide surveys on cosymplectic geometry due to their lack \cite{Cape,LeonChineaMarrero}.
Our particular interest in cosymplectic structures resides in their use in the description of time-dependent mechanics. They are present
in numerous formulations of classical regular lagrangians \cite{LeonMarinMarrero}, hamiltonian systems \cite{Laci} or Tulczyjew-like descriptions \cite{LeonMarrero} in terms
of lagrangian submanifolds \cite{LMM}.

However, there are more written monographs on contact geometry. The interest of contact structures roots in their applications 
in partial differential equations appearing in thermodynamics \cite{rajeev}, geometric mechanics \cite{IboLeonMarmo},
geometric optics \cite{CariNasarre,hamil1,hamil2}, geometric quantization \cite{rajeev} and applications to low dimensional topology, as it can be the characterization
of Stein manifolds \cite{Otto,Stein}.
Also, the theory of contact structures is linked to many other geometric backgrounds, as symplectic geometry, riemannian and complex geometry,
analysis and dynamics \cite{Blair, Godbillon}.

For both cosymplectic and contact approaches, the role of a vector field (said to be hamiltonian) with a corresponding smooth function (the hamiltonian function) with respect
to its corresponding structure, is primordial to have dynamics. Furthermore, this vector field will be key in the construction of a geometric Hamilton--Jacobi theory.
This theory has grown popular due to its simplicity and its equivalence to other theories of classical mechanics.
It is based on a principal idea: a hamiltonian vector field $X_{H}$ can be projected into the configuration manifold by means of a 1-form $dW$, then the integral curves of the projected
vector field $X_{H}^{dW}$ can be transformed into integral curves of $X_{H}$ provided that $W$ is a solution of the Hamilton--Jacobi equation \cite{Arnold,Gold,Kibble,LL,LiberMarle,Rund}. In the last decades, the Hamilton--Jacobi theory has been interpreted in modern geometric terms \cite{CGMMMLRR,CGMMMLRR1,LeonIglDiego,MarrSosa,Pauffer}
and has been applied in multiple settings: as nonholonomic \cite{CGMMMLRR,CGMMMLRR1,LeonIglDiego,LeonMarrDiego}, singular lagrangian mechanics \cite{LeonMarrDiegoVaq,LeonDiegoVaq2} and classical field theories \cite{LeonMarrDiego08, Cedric}.
The construction of a Hamilton--Jacobi theory often relies in the existence of lagrangian/legendrian submanifolds, a notion that has gained a lot of attention given its applications in dynamics since their introduction by Tulczyjew \cite{tulzcy1,tulzcy2}. We show how these submanifolds
are a necessary condition for the obtainance of particular solutions through a geometric Hamilton--Jacobi equation. We will devise
this fact by using particular cases of lagrangian/legendrian submanifolds in cosymplectic and contact geometry.

%
%
%

The paper is organized as follows: Section 2 is dedicated to review fundamentals on geometric mechanics and the geometric Hamilton--Jacobi equation.
In Section 3, we recall some remarkable geometric structures of importance in mechanics. In particular, we focus on dynamics explained in geometric terms through contact and cosymplectic manifolds.
Section 4 contains the theory of lagrangian--legendrian submanifolds which will be key in the formulation of the Hamilton--Jacobi theory in subsequent sections.
In Section 5, we propose a geometric Hamilton--Jacobi theory on cosymplectic manifolds and illustrate our result with examples: a one-dimensional
trigonometric system and two-dimensional nonlinear oscillators. One of the oscillators is the well-known Winternitz--Smorodinsky oscillator, for which we obtain an explicit expression
for the solution $\gamma$ of the Hamilton--Jacobi equation on the cosymplectic manifold.
Similarly, we devote Section 6 to a geometric Hamilton--Jacobi equation on a contact manifold. We also illustrate our result through an example,
a unidimensional damped oscillator with a dissipative term. 




\section{Geometric Mechanics: Fundamentals}
We hereafter assume all mathematical objects to be $C^{\infty}$, globally defined and that all manifolds are connected. This permit us to
omit technical details while highlighting the main aspects of our theory.


\subsection*{Hamiltonian Mechanics}

A classical hamiltonian system is given by a Hamilton function $H(q^i,p_i)$, where $(q^i)$ are the positions in a configuration manifold $Q$ and $(p_i)$ are the conjugated
momenta, for $i=1,\dots,n$. The hamiltonian can be interpreted as total energy of the system $H=T+V$, where $H$ is a function on the cotangent bundle $T^{*}Q$ of $Q$.
We compute the differential of the Hamiltonian function, 
$$dH=\sum_{i=1}^n \left(\frac{\partial H}{\partial q^i}dq^i+\frac{\partial H}{\partial p_i}dp_i\right)$$
 and write the equation
\begin{eqnarray}\label{symmatrix}
\left. X_H = \left( \begin{array}{ccc}
0&I_n\\
-I_n&0 \end{array} \right) \left(
\begin{array}{c}
\frac{\partial H}{\partial q^i}\\
\frac{\partial H}{\partial p_i}
\end{array}
\right)
\right.
\end{eqnarray}
where $I_n$ is the identity matrix of order $n$. The above matrix is called a symplectic matrix.
The vector field $X_H$ is called a {\it hamiltonian vector field} and its integral curves $(q^i(t),p_i(t))$ satisfy the Hamilton equations 
\begin{equation}\label{hamileq7}
\left\{\begin{aligned}
 {\dot q}^i&=\frac{\partial H}{\partial p_i},\\
 {\dot p}_i&=-\frac{\partial H}{\partial q^i}
 \end{aligned}\right.
 \end{equation}
for all $i=1,\dots,n.$
We can define a Poisson bracket of two functions as 
$$\{f,g\}=\sum_{i=1}^n\left(\frac{\partial f}{\partial q^i}\frac{\partial g}{\partial p_i}-\frac{\partial f}{\partial p_i}\frac{\partial g}{\partial q^i}\right),$$
which is bilinear, skew symmetric and fulfills the Jacobi identity
\begin{equation*}
 \{f,\{g,h\}\}+\{f,\{g,h\}\}+\{h,\{f,g\}\}=0,\quad \forall f,g,h\in C^{\infty}(Q).
\end{equation*}
\noindent
The symplectic two-form 
\begin{equation}\label{symp}
\omega_Q=\sum_{i=1}^n dq^i\wedge dp_i
\end{equation}
 has the associated symplectic matrix explained in \eqref{symmatrix}. 
We can rewrite the Hamilton equations \eqref{hamileq7} in a compact, geometric form
\begin{equation}
 \iota_{X_H}\omega_Q=dH,
\end{equation}
\noindent
where $X_H$ is the Hamiltonian vector field whose expression in coordinates is
\begin{equation}
 X_H=\sum_{i=1}^n \left(\frac{\partial H}{\partial p_i}\frac{\partial}{\partial q^i}-\frac{\partial H}{\partial q^i}\frac{\partial}{\partial p_i}\right).
\end{equation}

The pair $(T^{*}Q, \omega_Q)$ is the prototype for any symplectic manifold as we show in subsequent lines.
A {\it symplectic manifold} is a pair $(M,\omega)$ such that the two-form $\omega$ is regular (that is, $\omega^n\neq 0$) and closed. Then, $M$ is even dimensional, say $2n$. 
The Darboux theorem states that given a symplectic manifold $(M,\omega)$ we can find Darboux coordinates $(q^i,p_i)$ around each point of $M$ such that the symplectic form
is \eqref{symp}.
Indeed, any symplectic manifold is locally equivalent to the cotangent bundle $T^{*}Q$ of a configuration manifold $Q$.

Given a configuration manifold $Q$, its cotangent bundle $T^{*}Q$ is the phase space.
We consider the canonical projection $\pi_Q:T^{*}Q\rightarrow Q$. 
From the Poisson bracket, we can define a canonical two-contravariant tensor such that $\Lambda_Q(df,dg)=\{f,g\}$, for all $f,g\in C^{\infty}(T^{*}Q)$.
This is what we call a Poisson bivector. In Darboux coordinates it reads
\begin{equation}
\Lambda_Q=\sum_{i=1}^n \frac{\partial}{\partial q^i}\wedge \frac{\partial}{\partial p_i},
\end{equation}
It is the contravariant version of the symplectic form \eqref{symp}. Furthermore, we consider the so-called Liouville form $\theta_Q=p_idq^i$ on $T^{*}Q$ such that $\omega_Q=-d\theta_Q$.

\subsection*{Hamilton--Jacobi equation}

The Hamilton equations \eqref{hamileq7} can be equivalently be solved with the aid of the Hamilton--Jacobi theory.
It consists of finding a {\it principal function}  $S(t,q^i)$, that fulfills
\begin{equation}\label{tdepHJ}
 \frac{\partial S}{\partial t}+H\left(q^i,\frac{\partial S}{\partial q^i}\right)=0,\quad i=1,\dots,n
\end{equation}
where $H=H(q^i,p_i)$ is the hamiltonian function of the system. Equation \eqref{tdepHJ} is referred to as the {\it Hamilton--Jacobi equation.}
 If we set the principal function to be separable in 
time, $S=W(q^1,\dots,q^n)-Et,$
where $E$ is the total energy of the system, then \eqref{tdepHJ} will now read \cite{AbraMars,Gold}

\begin{equation}\label{HJeq1}
 H\left({q}^i,\frac{\partial W}{\partial {q}^i}\right)=E.
\end{equation}
The Hamilton--Jacobi equation is a useful intrument to solve the Hamilton equations for $H$.
Indeed, if we find a solution $W$ of \eqref{HJeq1}, then any solution of the Hamilton
equations is retrieved by taking ${p}_i=\partial W/\partial {q}^i,$ in case we can provide a complete solution.

Geometrically, the Hamilton--Jacobi theory can be reformulated as follows. Given a hamiltonian
vector field $X_{H}:T^{*}Q\rightarrow TT^{*}Q$ and a one-form $dW$, we define the {\it projected}\footnote{By projected we do not refer to a projective vector field but to the restriction of a hamiltonian vector field on the phase space $T^{*}Q$ along the image of $dW$.}
vector field $X_{H}^{dW}:Q\rightarrow TQ$. Then, the integral curves of $X_{H}^{dW}$ can be transformed into integral curves of $X_{H}$ provided that $W$ is a solution of \eqref{HJeq1}. 
This explanation can be represented by the following diagram

\[
\xymatrix{ T^{*}Q
\ar[dd]^{\pi} \ar[rrr]^{X_H}&   & &TT^{*}Q\ar[dd]^{T\pi}\\
  &  & &\\
 Q\ar@/^2pc/[uu]^{dW}\ar[rrr]^{X_H^{dW}}&  & & TQ}
\]

\bigskip

This implies that $(dW)^{*}H=E$, with $dW$ a section of the cotangent bundle. In other words, we are looking for a section $\alpha$ of $T^{*}Q$ 
such that $\alpha^{*}H=E$. As it is well-known, the image of a one-form is a lagrangian submanifold of $(T^{*}Q, \omega_Q)$ if and only if $d\alpha=0$ \cite{AbraMars}.
That is, $\alpha$ is locally exact, say $\alpha=dW$ on an open subset around each point.


\section{Geometric structures and dynamics}
A {\it Jacobi structure} is the triple $(M,\Lambda, Z)$, where $Z$ is a vector field and $\Lambda$ is a
skew-symmetric bivector and such that they fulfill the following integrability conditions
\begin{equation}\label{jaccond}
 [\Lambda,\Lambda]=2Z\wedge \Lambda,\qquad \mathcal{L}_Z\Lambda=0.
\end{equation}
 If we drop the integrability conditions, we say we have an {\it almost Jacobi manifold.} The bivector $\Lambda$ defines a {\it Jacobi bracket} \cite{QuanLeonMarrPad}
\begin{equation}\label{jacbracket}
 \{f,g\}=\Lambda(df,dg)+fZ(g)-gZ(f),\quad f,g,\in M
\end{equation}
that is skew-symmetric and satisfies the weaker Leibniz identity condition
\begin{equation}
 \text{supp}\{f,g\}\subseteq  \text{supp}\{f\}\cap  \text{supp}\{g\}.
\end{equation}
This implies that the Jacobi bracket \eqref{jacbracket} is not a derivation in each argument but satisfies the Jacobi identity if $(\Lambda,Z)$ is Jacobi.
The space $C^{\infty}(M,\mathbb{R})$ is a local Lie algebra in the Kirillov sense \cite{kiri}.

Consider the morphism $\sharp_{\Lambda}:\Omega^{1}(M)\rightarrow \mathfrak{X}(M)$ that is the $C^{\infty}(M,\mathbb{R})$ linear mapping induced by $\Lambda$
between the $C^{\infty}$ modules of one-forms $\Omega^{1}(M)$ and vector fields $\mathfrak{X}(M)$ defined on $M$.
The sharp-lambda morphism $\sharp_{\Lambda}$ provides the pairing and components of the bivector $\langle \sharp_{\Lambda}(\alpha),\beta \rangle=\Lambda(\alpha,\beta)$, for $\alpha,\beta$ one-forms in $\Omega^{1}(M)$.
Vector fields associated with functions $f$ on the algebra of smooth functions $C^{\infty}(M,\mathbb{R})$ are defined as
\begin{equation*}
 X_f=\sharp_{\Lambda}(df)+fZ,
\end{equation*}

The {\it characteristic distribution} $\mathcal{C}$ of $(M,\Lambda,Z)$ is a subset of $TM$ generated by the
values of all the vector fields $X_f$. 
This characteristic distribution $\mathcal{C}$ is defined by $\Lambda$ and $Z$, that is,
\begin{equation}
 \mathcal{C}_p=\sharp_{\Lambda_p}(T^{*}_pM)+<Z_p>,\quad \forall p\in M
\end{equation}
where $\sharp_p:T_p^{*}M\rightarrow T_pM$ is the restriction of $\sharp_{\Lambda}$ to $T^{*}_pM$ for every $p\in M$.
Then, $\mathcal{C}_p=\mathcal{C}\cap T_{p}M$ is the vector subspace of $T_pM$ generated by
$Z_p$ and the image of the linear mapping $\sharp_p$. The distribution is said to be {\it transitive} if the characteristic distribution
is the whole tangent bundle $TM$.

\begin{definition}
 Given two Jacobi manifolds $(M_1,\Lambda_1,Z_1)$ and $(M_2,\Lambda_2,Z_2)$ we say that the map $\phi:{M_1}\rightarrow M_2$ is a {\it Jacobi map}
if given two functions $f,g \in C^{\infty}(M_2)$,
\begin{equation*}
 \{f\circ \phi,g\circ \phi\}_{M_1}=\{f,g\}_{M_2}\circ \phi.
\end{equation*}

\end{definition}

A Jacobi manifold $(M,\Lambda,Z)$ is said to be be Poisson when $Z=0$; in that case we write $(M,\Lambda)$ instead of $(M,\Lambda,0)$. 

Next, we shall show several examples of Poisson and Jacobi manifolds.

\subsection{Symplectic manifolds}

 Consider a pair $(M,\Omega)$, where $\Omega$ is a symplectic two-form. We define the map  
\begin{equation}\label{inv}
\flat:TM\rightarrow T^{*}M \quad \text{such that}\quad  \flat(X)=\iota_X\Omega
\end{equation}
which is an isomorphism. We can define its inverse as $\sharp:T^{*}M\rightarrow TM$. The bracket
\begin{equation}\label{corchet}
 \{f,g\}=\Omega(\sharp(df),\sharp(dg))=\langle dg,\sharp(df)\rangle=-\langle df, \sharp(dg)\rangle
\end{equation}
satisfies the Jacobi identity. The Hamiltonian vector field is $X_f=\sharp(df).$
The pair $(M,\Lambda)$ is a Poisson manifold of necessarily even dimension with
Poisson tensor given by
\begin{equation*}
\Lambda(\alpha,\beta)=\Omega(\sharp(\alpha),\sharp(\beta))
\end{equation*}
for $\alpha,\beta$ one-forms. In this case $\sharp=\sharp_{\Lambda}$ and $\sharp=\flat^{-1}$. 

\subsection{Locally conformally symplectic structures}

An {\it almost symplectic manifold} is a pair $(M,\Omega)$
where $\Omega$ is a nondegenerate two-form and $M$ is even dimensional. An almost symplectic manifold
is said to be {\it locally conformally symplectic} if for each point $x\in M$
there is an open neighborhood $U$ such that $d(e^{\sigma}\Omega)=0,$ for $\sigma:U\rightarrow \mathbb{R}$, so $(U,e^{\sigma}\Omega)$
is a symplectic manifold. If $U=M$, then it is said to be globally conformally symplectic.
An almost symplectic manifold is a locally (globally) conformally symplectic if there exists a one-form
$\eta$ that is closed $d\eta=0$ and
\begin{equation}
 d\Omega=\eta\wedge \Omega.
\end{equation}
The one-form $\eta$ is called the {\it Lee one-form}. Locally conformally symplectic manifolds (L.C.S.) with Lee form $\eta=0$ are symplectic manifolds.
We define a bivector $\Lambda$ on $M$ and a vector field $Z$ given by
\begin{equation}
 \Lambda(\alpha,\beta)=\Omega(\flat^{-1}(\alpha),\flat^{-1}(\beta))=\Omega(\sharp(\alpha),\sharp(\beta)),\quad Z=\flat^{-1}(\eta)
\end{equation}
with $\alpha,\beta\in \Omega^{1}(M)$ and $\flat: \mathfrak{X}(M)\rightarrow \Omega^{1}(M)$ is the isomorphism of $C^{\infty}(M,\mathbb{R})$
modules defined by $\flat(X)=\iota_X\Omega$. Here $\sharp=\flat^{-1}$. In this case, we also have $\sharp_{\Lambda}=\sharp$. The vector field $Z$ satisfies $\iota_{Z}\eta=0$ and $\mathcal{L}_{Z}\Omega=0, \mathcal{L}_{Z}\eta=0$.
Then, $(M,\Lambda,Z)$ is an even dimensional Jacobi manifold. 
There is a classical Darboux theorem that states the following. Around every point $x\in M$, there exist coordinates and a local function $\sigma$ 
such that $d\sigma=\eta$. Then, 
\begin{equation}
\Omega=e^{\sigma}\sum_{i=1}^n dq^i\wedge dp_i, \quad \eta=d\sigma=\sum_{i=1}^n \left(\frac{\partial \sigma}{\partial q^i}dq^i+\frac{\partial \sigma}{\partial p_i}dp_i\right)
\end{equation}
and, in consequence, we have
\begin{equation}
 \Lambda=e^{-\sigma}\sum_{i=1}^n \left(\frac{\partial}{\partial q^i}\wedge \frac{\partial}{\partial p_i}\right),\quad Z=e^{-\sigma}\sum_{i=1}^n\left(\frac{\partial \sigma}{\partial p_i}\frac{\partial}{\partial q^i}-\frac{\partial \sigma}{\partial q^i}\frac{\partial}{\partial p_i}\right).
\end{equation}
The Hamiltonian vector field corresponding with the function $f$ is
\begin{equation}
 X_f=\sharp_{\Lambda}(df)+fZ
\end{equation}

\subsection{Almost cosymplectic structure}

 An almost cosymplectic manifold is a $2n+1$-dimensional manifold $M$ equi\- pped with $(\eta,\Omega)$, where $\eta$ is a one-form and $\Omega$ is a two-form such that $\eta\wedge \Omega^n\neq 0$.
Therefore, we have an isomorphism of $C^{\infty}$-modules $\flat:\mathfrak{X}(M) \rightarrow \Omega^{1}(M)$ defined by
\begin{equation}\label{betamorph}
 \flat(X)=i_X\Omega+\eta(X)\eta.
\end{equation}
\begin{theorem}
 If $(M,\eta,\Omega)$ is an almost cosymplectic structure, then there exists a unique vector field $\mathcal{R}$, the so-called Reeb vector field such that
\begin{equation}\label{reebeq}
 \iota_{\mathcal{R}}\eta=1,\quad \iota_{\mathcal{R}}\Omega=0.
\end{equation}
\end{theorem}
In other words, $\mathcal{R}=\flat^{-1}(\eta)$.
Now, we will consider two particular classes of almost cosymplectic manifolds.

\subsubsection{Cosymplectic structure}
An almost cosymplectic structure $(M,\Omega,\eta)$ is a {\it cosymplectic structure} if $d\eta=0$, $d\Omega=0$; recall that $\Omega^n\wedge \eta\neq 0$.
A cosymplectic manifold is equipped with the $\flat$ isomorphism in \eqref{betamorph} and the Reeb vector field is retrieved as $\mathcal{R}=\flat^{-1}(\eta)=\sharp(\eta)$
and satisfies \eqref{reebeq}. In this case $\sharp=\flat^{-1}$.
 A cosymplectic manifold is a particular case of odd dimensional Poisson manifold. It gives rise to a Poisson bivector given by
 \begin{equation*}
  \Lambda(\alpha,\beta)=\Omega(\sharp(\alpha),\sharp(\beta))
 \end{equation*}

There exist Darboux coordinates $\{t,q^i,p_i\}$ on $T^{*}Q\times \mathbb{R}$ with $i=1,\dots, n$ such that 

\begin{equation}\label{cosym}
 \Omega=\sum_{i=1}^n dq^i\wedge dp_i,\quad \eta=dt,
\end{equation}
and therefore,
\begin{equation}\label{cosym2}
 \mathcal{R}=\frac{\partial}{\partial t}, \quad \Lambda=\sum_{i=1}^n \frac{\partial}{\partial q_i}\wedge \frac{\partial}{\partial p_i}.
\end{equation}
To define a Poisson structure, we consider the bivector 
\begin{equation}
\Lambda(df,dg)=df(\sharp_{\Lambda}(dg))= \{f,g\}.
\end{equation}
In this case, $\sharp_{\Lambda}\neq \sharp$.


\subsubsection{Contact structure}

An almost cosymplectic structure $(M,\Omega,\eta)$ is a {\it contact structure} if $\Omega=d\eta.$ From here, we refer to $\eta$ as a contact form and $(M,\eta)$ a contact manifold.
It is satisfied that $\eta\wedge (d\eta)^n\neq 0$ for all $x\in M$.
%
%
%

A contact manifold is a Jacobi manifold whose associated bivector $\Lambda$ is given by
\begin{equation}\label{lambdacontacto}
 \Lambda(\alpha,\beta)=d\eta(\flat^{-1}(\alpha),\flat^{-1}(\beta))=d\eta(\sharp(\alpha),\sharp(\beta))
\end{equation}
for all $\alpha,\beta \in \Omega^{1}(M)$ and $\flat:\mathfrak{X}(M)\rightarrow \Omega^{1}(M)$ is the isomorphism given by
\begin{equation}
 \flat(X)=\iota_Xd\eta+\eta(X)\eta.
\end{equation}
Here $\sharp=\flat^{-1}$.
The Reeb vector field is $\mathcal{R}=\flat^{-1}(\eta)=\sharp(\eta)$ and it is the unique vector field that satisfies \eqref{reebeq}. Then $(M,\Lambda,\mathcal{R})$ is a Jacobi manifold. Then, the bracket on a contact manifold is defined as
\begin{equation}
 \{f,g\}=\Lambda(df,dg)+f\mathcal{R}(g)-g\mathcal{R}(f).
\end{equation}

There exist coordinates $\{t,q^i,p_i\}$, with $i=1,\dots,n$, such that
\begin{equation}\label{lambdaex}
 \eta=dt-\sum_{i=1}^n p_idq^i,\quad \Lambda=\sum_{i=1}^n \left(\frac{\partial}{\partial q^i}+p_i\frac{\partial}{\partial t}\right)\wedge \frac{\partial}{\partial p_i},\quad \mathcal{R}=\frac{\partial}{\partial t}.
\end{equation}

%
In this case, $\sharp_{\Lambda}=\sharp$.

\subsection*{Structure theorem for Jacobi manifolds}

\begin{theorem}
The characteristic distribution of a Jacobi manifold $(M,\Lambda,Z)$ is completely integrable in the sense of Stefan--Sussmann \cite{Stefan, Sussmann}, thus
$M$ defines a foliation whose leaves are not necessarily of the same dimension, and it is called a {\it characteristic foliation}. Each leaf has a unique
transitive Jacobi structure such that its canonical injection into $M$ is a Jacobi map.
Each leaf defines
\begin{enumerate}
 \item A locally conformally symplectic manifold if the dimension is even.
\item A manifold equipped with a contact one-form if its dimension is odd.
\end{enumerate}

\end{theorem}

\begin{table}[H]{\footnotesize
  \noindent
\caption{{\small {\bf Geometric structures}. The space $C^{\infty}(M,\mathbb{R})$ is a local Lie algebra in the Kirillov sense and the 
morphism $\sharp_{\Lambda}:\Omega^{1}(M)\rightarrow \mathfrak{X}(M)$ is a $C^{\infty}(M,\mathbb{R})$ mapping induced by $\Lambda$ between $C^{\infty}$ modules of $\Omega^{1}$ and $\mathfrak{X}$ on $M$. $\sharp_p$ is the restriction of $\sharp_{\Lambda}$ to $T^{*}_pM$.
The flat morphism  $\flat:\mathfrak{X}(M)\rightarrow \Omega^{1}(M)$ is a mapping between $C^{\infty}$ modules; it is generally defined as $\flat(X)=\iota_X\Omega+\eta(X)\eta$. Obviously, it reduces to $\flat(X)=\iota_X\Omega$ for particular cases. Here $\alpha$ and $\beta$ are one-forms on $\Omega^{1}(M)$.}}
\label{table2}
\medskip
\noindent\hfill
\resizebox{\textwidth}{!}{\begin{minipage}{\textwidth}
\begin{tabular}{ l l l l }
 
\hline
 &&\\[-1.5ex]
 Structure&  Characterization &Bracket and h.v.f.& Induced structure \\[+1.0ex]
\hline
 &  & \\[-1.5ex]

 {\bf L.C.S.}& $d\Omega=\eta\wedge \Omega$  & $\{f,g\}=\Lambda(df,dg)+fZ(g)-gZ(f) $ &  $\Lambda(\alpha,\beta)=\Omega(\sharp(\alpha),\sharp(\beta))$   \\[+1.0ex] $(U,\Omega)$& $\mathcal{L}_Z\Omega=0$ &   $X_h=\sharp_{\Lambda}(dh)+hZ$  & $Z=\sharp(\eta)$\\[+1.0ex] $d(e^{\sigma}\Omega)=0$ & $\mathcal{L}_Z\eta=0$& $\sharp=\flat^{-1},\quad \sharp_{\Lambda}=\sharp$ & $\mathcal{C}_p=\sharp_{\Lambda}(T_p^{*}M)+\langle Z_p\rangle$\\[+1.0ex] even dim & &  & Jacobi \\[+2.0ex] 
\hline\\[0.5ex]
 {\bf Contact} & $\Omega=d\eta$ &  $\{f,g\}=\Lambda(df,dg)+f\mathcal{R}(g)-g\mathcal{R}f$ &  $\Lambda(\alpha,\beta)=d\eta(\sharp(\alpha),\sharp(\beta))$ \\[+1.0ex]$(M,\eta)$  & $\eta\wedge (d\eta)^n\neq 0$  & $X_h=\sharp_{\Lambda}(dh)+h\mathcal{R}$  & $Z=\mathcal{R}=\sharp(\eta)$ \\[+1.0ex]odd dim &&$\sharp=\flat^{-1},\quad \sharp_{\Lambda}=\sharp$ & $\mathcal{C}_p=\sharp_{\Lambda}(T_p^{*}M)+\langle \mathcal{R}_p\rangle$ \\[+1.0ex]  & & & Jacobi\\[+2.0ex]
\hline\\[0.5ex]

 {\bf Cosymplectic} &  $\Omega^n\wedge \eta\neq 0$& $\{f,g\}=\Lambda(df,dg)=df(\sharp_{\Lambda}(dg))$ & $\Lambda(\alpha,\beta)=df(\sharp_{\Lambda}(dg))$ \\[+1.0ex] $(M,\Omega,\eta)$   & $d\eta=0,d\Omega=0$  &  $X_h=\sharp (dh)$ & $Z=0, \mathcal{R}=\sharp(\eta)$ \\[+1.0ex] odd dim & & $\sharp=\flat^{-1},\quad \sharp_{\Lambda}\neq \sharp$ &$\mathcal{C}_p=\sharp(T_p^{*}M)$\\[1.0ex] & & & Poisson \\[+2.0ex]
\hline\\[0.5ex]
 {\bf Symplectic} & $(M,\Omega)$ & $\{f,g\}=\Omega(\sharp(df),\sharp(dg))$&    $\Lambda(\alpha,\beta)=\Omega(\sharp(\alpha),\sharp(\beta))$  \\[+1.0ex] $(M,\Omega)$ & $d\Omega=0$  &$X_h=\sharp_{\Lambda}(dh)$ & $Z=0,\mathcal{R}=0$ \\[+1.0ex] even dim&  $\Omega^n\neq 0$ &$\sharp=\flat^{-1},\quad \sharp_{\Lambda}=\sharp$ &$\mathcal{C}_p=\sharp_{\Lambda}(T_p^{*}M)$ \\[+2.0ex] & & & Poisson\\[+2.0ex]
\hline \\[0.5ex]

    
\end{tabular}
  \end{minipage}}
\hfill}
\end{table}

\begin{table}[H]{\footnotesize
  \noindent
\caption{{\small {\bf  Almost cosymplectic structures}. We choose canonical coordinates $\{q^i,p_i,t\}$ on $T^{*}Q\times \mathbb{R}$   }}
\label{table2}
\medskip
\noindent\hfill
\resizebox{\textwidth}{!}{\begin{minipage}{\textwidth}
\begin{tabular}{ l l l l }

\hline
 &&\\[-1.5ex]
 Structure& One-form & Reeb & Bivector and h.v.f. \\[+1.0ex]
\hline
 &  & \\[-1.5ex]
{\bf Cosymplectic}&  & &   \\[+1.0ex] $(M,\Omega,\eta)$ & $\theta_H=p_idq^î-Hdt$ & $\mathcal{R}_H=\frac{\partial}{\partial t}+\sum_{i=1}^n \frac{\partial H}{\partial p_i}\frac{\partial}{\partial q^i}-$ &  $\Lambda=\sum_{i=1}^n\frac{\partial}{\partial q^i}\wedge\frac{\partial}{\partial p_i}$   \\[+1.0ex] $\Omega^n\wedge \eta\neq 0$  &  & $-\sum_{i=1}^n \frac{\partial H}{\partial q^i}\frac{\partial}{\partial p_i}$ & $X_H=\mathcal{R}_H$ \\[2.0ex]
\hline\\[0.5ex]

{\bf Contact}&   & &     \\[+1.0ex] $(M,\eta)$  & $\eta=dt-\sum_{i=1}^n p_idq^i$ &$\qquad \mathcal{R}=\frac{\partial}{\partial t}$  & $\Lambda=\left(\frac{\partial}{\partial q^i}+p_i\frac{\partial}{\partial t}\right)\wedge \frac{\partial}{\partial p_i}$ \\[1.0ex] $\eta\wedge (d\eta)^n\neq 0$&  &  & $X_H=\sum_{i=1}^n \left(p_i\frac{\partial H}{\partial p_i}-H\right)\frac{\partial}{\partial t}$  \\[+1.0ex] $\Omega=d\eta$& & & $-\sum_{i=1}^n\left(p_i\frac{\partial H}{\partial t}+\frac{\partial H}{\partial q^i}\right)\frac{\partial}{\partial p_i}+\frac{\partial H}{\partial p_i}\frac{\partial}{\partial q^i}$ \\[2.0ex]
\hline\\

    
\end{tabular}
  \end{minipage}}
\hfill}
\end{table}

\section{Lagrangian--legendrian submanifolds}

Let $(M,\Lambda,Z)$ be a Jacobi manifold with characteristic distribution $\mathcal{C}$.

\begin{definition}
 A submanifold $N$ of a Jacobi manifold $(M,\Lambda,Z)$ is said to be a {\it lagrangian-legendrian submanifold} if the following equality holds
\begin{equation}
 \sharp (TN^{\circ})=TN \cap \mathcal{C},
\end{equation}
where $TN^{\circ}$ denotes the annihilator of $TN$.

\end{definition}

If $(M,\Lambda)$ is a Poisson manifold, the lagrangian-legendrian submanifold of $M$ will simply be called lagrangian.

\subsubsection*{Particular cases}
\begin{enumerate}
\item A submanifold $N$ of a symplectic manifold $(M,\Omega)$ is {\it lagrangian} if
\begin{equation}
 \sharp_{\Lambda}(TN^{\circ})=TN.
\end{equation}

 \item As a consequence, we deduce that a submanifold $N$ of a cosymplectic manifold $(M,\eta,\Omega)$ is {\it lagrangian} if
\begin{equation}
 \sharp_{\Lambda} (TN^{\circ})=TN\cap \mathcal{C}.
\end{equation}

\item A submanifold $N$ of a contact manifold $(M,\eta)$ is {\it legendrian} if the following condition is fulfillled
\begin{equation}
 \sharp_{\Lambda} (TN^{\circ})=TN.
\end{equation}
\end{enumerate}
\noindent
The following result gives a characterization of legendrian submanifolds of contact manifolds.

\begin{proposition}
A submanifold $N$ of a contact manifold $(M,\eta)$ is a {\it legendrian} submanifold if and only if it is an integral manifold of maximal dimension $n$ of the
distribution $\eta=0$. In this case, $\mathcal{C}$ is the whole tangent space to $N$.
\end{proposition} 

\begin{proof}
 Assume that $M$ has dimension $2n+1$. If a submanifold $N$ of $M$ is legendrian then the condition
\begin{equation*}
 \sharp_{\Lambda}(TN^{\circ})=TN
\end{equation*}
implies that $\eta|_{N}=0$. Moreover, $N$ has necessary dimension $n$, since $T_xN$ will be a lagrangian subspace
of the symplectic vector space $(\ker{\eta}_x,(d\eta)_x)$ for all $x\in N$. The converse is proved reversing the arguments.

\end{proof}

%
%

\section{Hamilton--Jacobi theory on cosymplectic manifolds}

\subsection{Geometric approach}
Consider the extended phase space $T^{*}Q\times \mathbb{R}$ and its canonical projections of the first and second factor, $\rho:T^{*}Q\times \mathbb{R}\rightarrow T^{*}Q$ and  $t:T^{*}Q\times \mathbb{R}\rightarrow \mathbb{R}$, respectively
 and a time-dependent hamiltonian $H:T^{*}Q\times \mathbb{R}\rightarrow \mathbb{R}$. Let us depict the problem with a diagram

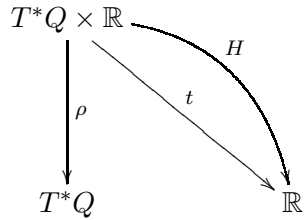
\[
\xymatrix{ T^{*}Q\times \mathbb{R}
\ar[dd]^{\rho}  \ar[ddrr]^{t}  \ar@/^2pc/[ddrr]^{H}  &   &\\
  &  & &\\
T^{*}Q & & \mathbb{R}}
\]

\bigskip
\noindent
We have canonical coordinates $\{q^i,p_i,t\}$ with $i=1,\dots,n$, where $(q^i,p_i)$ are fibered coordinates in $T^{*}Q$ and $t\in \mathbb{R}$.
We consider the two-form on $T^{*}Q\times \mathbb{R}$ as $\Omega_H=-d\theta_H$
and 
\begin{equation}\label{thetah}
\theta_H=\theta_Q-Hdt
\end{equation}
where $\theta_Q$ is the canonical Liouville one-form. 
We abuse notation by identifying the pullbacks of the one-forms with the one-forms themselves. That is, $\rho^{*}(\theta_Q)=\theta_Q$.
Hence, 
\begin{equation}\label{oh}
\Omega_H=\sum_{i=1}^n dq^i\wedge dp_i+dH\wedge dt.
\end{equation}
Let us consider the cosymplectic structure $(dt,\Omega_H).$
The corresponding Reeb vector field needs to satisfy 
\begin{equation}\label{reebc}
\iota_{\mathcal{R}_H}dt=1,\quad \iota_{\mathcal{R}_H}\Omega_H=0.
\end{equation}
The unique Reeb vector field that satisfies \eqref{reebc} has the following expression in coordinates
\begin{equation}\label{reebcosym}
 \mathcal{R}_H=\frac{\partial}{\partial t}+\sum_{i=1}^n\frac{\partial H}{\partial p_i}\frac{\partial}{\partial q^i}-\sum_{i=1}^n \frac{\partial H}{\partial q^i}\frac{\partial}{\partial p_i}.
\end{equation}

The corresponding classical Hamilton--Jacobi equations are
\begin{equation}\label{hamileq22}
\left\{\begin{aligned}
 {\dot q}^i&=\frac{\partial H}{\partial p_i},\\
 {\dot p}_i&=-\frac{\partial H}{\partial q^i},\qquad \forall i=1,\dots,n. \\
{\dot t}&=1.
 \end{aligned}\right.
 \end{equation}
\noindent
Since $\dot{t}=1$, we can consider $t$ a time-parameter (up to an affine change).

\noindent
We consider the fibration
$\pi: T^{*}Q\times \mathbb{R}\rightarrow Q\times \mathbb{R}$ and a section $\gamma$ of $\pi:T^{*}Q\times \mathbb{R} \rightarrow Q\times \mathbb{R}$, i.e., $\pi\circ \gamma=\text{id}_{Q\times \mathbb{R}}$. 
Also, we assume that $\text{Im}(\gamma_t)$ with $\gamma_t:Q\rightarrow T^{*}Q\times \mathbb{R}$ such that
$\gamma_t(q^i)$ in coordinates $(q^i,\gamma^i(q^i,t),t)$ is a lagrangian submanifold for a fixed time $t$ of the cosymplectic manifold $(T^{*}Q\times \mathbb{R},dt,\Omega_H)$ for a fixed time, that is $d\gamma_t=0$.
%

\[
\xymatrix{
 T^{*}Q\times \mathbb{R}\ar[rr]^{\rho} & &  T^{*}Q\ar[rr]^{\pi} & &  Q\times \mathbb{R}\ar[rr] & &  Q\ar@/^2pc/[llllll]^{\gamma_t}}
\]

We can use $\gamma$ to project $\mathcal{R}_H$ on $Q\times \mathbb{R}$
just defining a vector field $\mathcal{R}^{\gamma}_H$, the denominated {\it projected} vector field on $Q\times \mathbb{R}$ by
\begin{equation}\label{hjr}
 \mathcal{R}^{\gamma}_H=T_{\pi}\circ \mathcal{R}_H\circ \gamma
\end{equation}
The following diagram summarizes the above construction
\[
\xymatrix{ T^{*}Q\times \mathbb{R}
\ar[dd]^{\pi} \ar[rrr]^{\mathcal{R}_H}&   & &T(T^{*}Q\times \mathbb{R})\ar[dd]^{T_{\pi}}\\
  &  & &\\
Q\times \mathbb{R} \ar@/^2pc/[uu]^{\gamma}\ar[rrr]^{\mathcal{R}^{\gamma}_H}&  & & T(Q\times \mathbb{R})}
\]
\begin{definition}
 If $\alpha$ is a one-form, locally expressed as $\alpha=\sum_{i=1}^n\alpha_i dq^i$, we designate by $\alpha^{V}$ the {\it vertical lift} \cite{yanoishi}
or vector fields associated with $\alpha$, defined by
\begin{equation*}
 \iota_{\alpha^V}\omega_Q=\alpha
\end{equation*}
Hence, the vector field $\alpha^V$ has the local expression
\begin{equation}
\alpha^{V}=-\sum_{i=1}^n\alpha_i\frac{\partial}{\partial p_i}.
\end{equation}
\end{definition}

\begin{theorem}
 The vector fields $\mathcal{R}_H$ and $\mathcal{R}^{\gamma}_H$ are $\gamma$-related if and only if the following equation is satisfied
\begin{equation}\label{eqtheorem2}
[d(H\circ \gamma_t)]^{V}=\dot{\gamma}_q
\end{equation}
where $[\dots]^{V}$ denotes the vertical lift of a one-form on $Q$ to $T^{*}Q$.  Now $\dot{\gamma}_q$ is the tangent vector in a point $q$ associated with the curve
\[
\xymatrix{
 \mathbb{R}\ar[rr]\ar@/^2pc/[rrrrrr]^{\gamma_q} & &  Q\times \mathbb{R}\ar[rr] & &  T^{*}Q\times \mathbb{R}\ar[rr]^{\rho} & &  T^{*}Q}
\]
Notice that these applications are given for a fixed point $q\rightarrow (q,t,\gamma)$.

\end{theorem}
\begin{proof}
 The vector fields $\mathcal{R}_H$ and $\mathcal{R}^{\gamma}_H$ are $\gamma$ related if $T\gamma(\mathcal{R}^{\gamma}_H)=\mathcal{R}_H$. That is,
\begin{equation}\label{transf11}
T\gamma(\mathcal{R}^{\gamma}_H)=T\gamma\left(\frac{\partial}{\partial t}+\sum_{i=1}^n\frac{\partial H}{\partial p_i}\frac{\partial}{\partial q^i}\right)
\end{equation}
We choose a section $\gamma(q^i,t)$ in coordinates $(q^i,\gamma^j(q^i,t),t)$ with $i,j=1,\dots,n$ such that the lift in the tangent bundle reads,
\begin{equation}\label{tangent2}
T\gamma\left(\frac{\partial}{\partial t}\right)=\frac{\partial}{\partial t}+\sum_{j=1}^n\frac{\partial \gamma^j}{\partial t}\frac{\partial}{\partial p_j},\quad T\gamma\left(\frac{\partial}{\partial q^i}\right)=\frac{\partial}{\partial q^i}+\sum_{j=1}^n\frac{\partial \gamma^j}{\partial q^i}\frac{\partial}{\partial p_j} 
\end{equation}
Introducing equations \eqref{tangent2} in equation \eqref{transf11}, it is straightforward to retrieve condition \eqref{eqtheorem2} if we use that 
$\gamma_t$ is closed. The closedness condition is necessary for the permutation of indices in intermediate steps to obtain \eqref{eqtheorem2}.
\end{proof}
Equation \eqref{eqtheorem2} is known as a {\it Hamilton--Jacobi equation on a cosymplectic manifold.}
In local coordinates $\{q^i,p_i,t\}$, we have
\begin{equation}\label{lc}
 \frac{\partial \gamma^j}{\partial t}+\sum_{i=1}^n \frac{\partial H}{\partial p_i}\frac{\partial \gamma^j}{\partial q^i}+\frac{\partial H}{\partial q^j}=0.
\end{equation}

\subsection{Complete solutions}

\begin{definition}
 A {\it complete solution} of the Hamilton--Jacobi equation on a cosymplectic manifold $(M,\eta,\Omega)$ 
 is a diffeomorphism $\Phi:Q\times \mathbb{R}\times \mathbb{R}^n\rightarrow T^{*}Q\times \mathbb{R}\times \mathbb{R}^n$ such that for a set of
 parameters $\lambda\in \mathbb{R}^n, \lambda=(\lambda_1,\dots,\lambda_n)$, the mapping
 
 \begin{equation}\label{compsolHJ}
 \begin{array}{ccc}
  \Phi_{\lambda}:Q\times \mathbb{R}& \rightarrow &  T^{*}Q\times \mathbb{R}  \\
  \Phi_{\lambda}(q,t) &\mapsto &  \Phi(q,\gamma(q,t),t)
 \end{array}
 \end{equation}
\noindent
is a solution of the Hamilton--Jacobi equation.

\medskip

We have the following diagram

\[
\xymatrix{ Q\times \mathbb{R}\times \mathbb{R}^n
\ar[dd]^{\alpha} \ar[rrr]^{\Phi}&   & &T^{*}Q\times \mathbb{R}\ar[dd]^{f_i}\ar[lll]^{\Phi^{-1}}\\
  &  & &\\
 \mathbb{R}^n \ar[rrr]^{\pi_i}&  & & \mathbb{R}}
\]

%

\noindent
with $\pi_i:\mathbb{R}^n\rightarrow \mathbb{R}$ the projection of $(\lambda_1,\dots,\lambda_n)$ to $\lambda_i$.
We define functions $f_i$ such that for a point $p\in T^{*}Q\times \mathbb{R}$, it is satisfied
\begin{equation}\label{functions8}
 f_i(p)=\pi_i\circ \alpha\circ \Phi^{-1}(p).
\end{equation}
and $\alpha:Q\times \mathbb{R}\times \mathbb{R}^n\rightarrow \mathbb{R}^n$ is the canonical projection.

\begin{theorem}
 If $\Phi$ is a complete solution of the Hamilton--Jacobi problem on a cosymplectic manifold, then the functions defined in \eqref{functions8} commute
 with respect to a Poisson bracket, that is,
 \begin{equation}
  \{f_i,f_j\}=0,\quad \forall i,j=1,\dots, n.
 \end{equation}

\end{theorem}

\begin{proof}
 It is immediate that
 \begin{equation}
  \text{Im}(\Phi_{\lambda})=\cap_{i=1}^n f_i^{-1}(\lambda_i)
 \end{equation}
 An element in $\text{Im}(\Phi_{\lambda})$ will be $\Phi_{\lambda}(x)$, for a point $x\in Q\times \mathbb{R}$ and it happens
 \begin{equation*}
  f_i(\Phi_{\lambda}(x))=f_i(\Phi(x,\lambda))=\lambda_i.
 \end{equation*}
 If $f_i$ is constant on $\text{Im}(\Phi_{\lambda})$, then $df_i$ vanishes on $T(\text{Im}(\Phi_{\lambda}))$. 
 If $\Phi_{\lambda}$ is a solution of the Hamilton--Jacobi equation, then  $\text{Im}\Phi_{\lambda}$ is a lagrangian submanifold and we have that
 \begin{equation*}
  \sharp_{\Lambda} (T(\text{Im}(\Phi_{\lambda})))^{\circ}=T(\text{Im}(\Phi_{\lambda}))\cap \mathcal{C}
 \end{equation*}
 that implies 
 \begin{equation}
  \{f_i,f_j\}=0,\quad \forall i,j=1,\dots,n.
 \end{equation}
\noindent
Because of the definition of the bracket $\{f_i,f_j\}=df_i(\sharp_{\Lambda}(df_j))$ and the definition of $\Lambda$ in \eqref{cosym2}, we deduce
that it is zero since $df_i$ is in the annihilator of $T(\text{Im}(\Phi_{\lambda}))$ and $\sharp_{\Lambda}(df_j)$ is in $T(\text{Im}(\Phi_{\lambda}))\cap \mathcal{C}$.

\end{proof}

\end{definition}

\subsection{Examples}

\subsubsection*{A trigonometric system}
Let us consider the time-dependent hamiltonian on $T^{*}Q\times \mathbb{R}$ with the local set of coordinates $\{q,p,t\}$
\begin{equation}
 H=\frac{p^2}{2}+\frac{q^2}{2}+\alpha \sin{(wt)}\frac{q^2p^2}{2}.
\end{equation}
In our setting, we consider the cosymplectic manifold $(T^{*}Q\times \mathbb{R},\Omega_H, \theta_H)$ where $\Omega_H$ and $\theta_H$ are
those given in \eqref{oh} in \eqref{thetah}, correspondingly.
The Reeb vector field according to the conditions \eqref{reebeq} has the expression
\begin{equation}
 \mathcal{R}_H=\frac{\partial}{\partial t}+\left(p+\alpha\sin{(wt)}q^2p\right)\frac{\partial}{\partial q}-\left(q+\alpha\sin{(wt)}p^2q\right)\frac{\partial}{\partial p}.
\end{equation}
We choose a lagrangian section $\gamma(q,t)$ whose components are $(q,\gamma(q,t),t)$.
The $\mathcal{R}_H^{\gamma}$ field on $Q\times \mathbb{R}$ is
\begin{equation}
 \mathcal{R}^{\gamma}_H=\frac{\partial}{\partial t}+\left(p+\alpha\sin{(wt)}q^2p\right)\frac{\partial}{\partial q}. 
\end{equation}
If we impose \eqref{hjr} to be fulfillled, we need to compute the terms
\begin{equation}
 T\gamma\left(\frac{\partial}{\partial t}\right)=\frac{\partial}{\partial t}+\frac{\partial \gamma}{\partial t}\frac{\partial}{\partial p},\quad T\gamma\left(\frac{\partial}{\partial q}\right)=\frac{\partial}{\partial q}+\frac{\partial \gamma}{\partial q}\frac{\partial}{\partial p},
\end{equation}
and the arising equation reads
\begin{equation}
 \frac{\partial \gamma}{\partial t}+\left(p+\alpha\sin{(wt)}q^2p\right)\frac{\partial \gamma}{\partial q}=q+\alpha\sin{(wt)}p^2q.
\end{equation}
This equation is a quasi-linear first-order PDE for a function $\gamma(q,t)$. It can be solved with the aid of the method of characteristics \cite{evans}
\begin{equation}
 dt=\frac{dq}{p+\alpha \sin{(wt)}q^2p}=\frac{d\gamma}{q+\alpha \sin{(wt)}p^2q}
\end{equation}
which turns in the following system of equations
\begin{align}\label{sysgp}
 \frac{dq}{dt}&=p(1+\alpha \sin{(wt)}q^2),\nonumber\\
\frac{d\gamma}{dt}&=q(1+\alpha \sin{(wt)}p^2).
\end{align}
Integrating the equations along the section $\gamma$, we have that $p=\gamma$, then, we can solve system \eqref{sysgp} whose
solutions result in
\begin{equation}\label{compsol}
 \gamma(q,t)=\frac{q}{\tanh{(t+C)}}
\end{equation}
where $C$ is a constant of integration.

Equation \eqref{compsol} is a particular solution of the Hamilton--Jacobi equation corresponding with a nonlinear trigonometric system on a cosymplectic manifold.

For the complete solution, we need $\gamma(q,t)$ expressed in terms of one single parameter $C$ as in \eqref{compsol}, according to the theory explained in \eqref{compsolHJ}.
We construct the diffeomorphism that provides the complete solution
 \begin{equation}\label{compsolHJtrig}
 \begin{array}{ccc}
  \Phi:\mathbb{R}\times \mathbb{R}\times\mathbb{R}& \rightarrow &  T^{*}\mathbb{R}\times \mathbb{R}  \\
  \Phi(q,t,C) &\mapsto &  \Phi(q,\gamma(q,t,C),t)
 \end{array}
 \end{equation}
 such that
 \begin{equation}
 \Phi(q,t,C)= \left(q,t,\frac{q}{\tanh{(t+C)}}\right)
 \end{equation}

\subsubsection*{Nonlinear oscillators}
Fris {\it et al} \cite{fris} studied in 1965 systems that admit separability in two different coordinates and obtained families of superintegrable
potentials with constants of motion linear or quadratic in the velocities (momenta). The two first families can be considered as the more
general euclidean deformations \cite{cr} with strengths $k_2$ and $k_3$ 

\begin{equation}\label{firstpotential}
 V_a=\frac{1}{2}\omega_0^2(x^2+y^2)+\frac{k_2}{x^2}+\frac{k_3}{y^2},\quad V_b=\frac{1}{2}\omega_0^2(4x^2+y^2)+k_2x+\frac{k_3}{y^2}
\end{equation}
\noindent
of the $1:1$ and $2:1$ harmonic oscillators preserving quadratic superintegrability. The superintegrability of $V_a$ is known as
Winternitz--Smorodinsky oscillator \cite{ws} studied by Evans \cite{evans2} for the general case of degrees of freedom.
For models in two dimensions, we have to choose two sections $\gamma(x,t), \gamma(y,t):Q\times \mathbb{R}\rightarrow T^{*}Q\times \mathbb{R}$ as a solution of the Hamilton--Jacobi equation.

{\bf (1)} For the first potential $V_a$, the hamiltonian reads
\begin{equation}
 H=\frac{1}{2}p_x^2+\frac{1}{2}p_y^2+\frac{1}{2}\omega_0^2(x^2+y^2)+\frac{k_2}{x^2}+\frac{k_3}{y^2}.
\end{equation}
In this case, the Reeb vector field according to \eqref{reebeq} is
\begin{equation}
 \mathcal{R}_H=\frac{\partial}{\partial t}+p_x\frac{\partial}{\partial x}+p_y\frac{\partial}{\partial y}-\left(\omega_0^2x-\frac{2k_2}{x^3}\right)\frac{\partial}{\partial p_x}-\left(\omega_0^2y-\frac{2k_3}{y^3}\right)\frac{\partial}{\partial p_y}
\end{equation}
\noindent
The {\it projected} vector field on $Q\times \mathbb{R}$ is
\begin{equation}
  \mathcal{R}_H^{\gamma}=\frac{\partial}{\partial t}+p_x\frac{\partial}{\partial x}+p_y\frac{\partial}{\partial y}
\end{equation}
Now,
\begin{align*}
 T\gamma (R_H^{\gamma})=\frac{\partial}{\partial t}+\frac{\partial \gamma^{[x]}}{\partial t}\frac{\partial}{\partial p_x}&+\frac{\partial \gamma^{[y]}}{\partial t}\frac{\partial}{\partial p_y}+\\
 &+p_x\left(\frac{\partial}{\partial x}+\frac{\partial \gamma^{[x]}}{\partial x}\frac{\partial}{\partial p_x}\right)+p_y\left(\frac{\partial}{\partial y}+\frac{\partial \gamma^{[y]}}{\partial y}\frac{\partial}{\partial p_y}\right)
\end{align*}
that compared with $\mathcal{R}_H$, it results in the Hamilton--Jacobi equations
\begin{align}\label{MPE}
  &\gamma_t^{[x]}+\frac{1}{m}\gamma^{[x]}\gamma_x^{[x]}=-\left(\omega_0^2x-\frac{2k_2}{x^3}\right),\nonumber\\
  &\gamma_t^{[y]}+\frac{1}{m}\gamma^{[y]}\gamma_y^{[y]}=-\left(\omega_0^2y-\frac{2k_3}{y^3}\right).
\end{align}
The two equations for the sections $\gamma^{[x]}$ and $\gamma^{[y]}$ are the same. They are quasilinear PDEs that
can be solved with the method of the characteristics. So, the associated characteritic system for $\gamma^{[x]}$ is
\begin{equation}
 dt=\frac{dx}{\gamma^{[x]}}=\frac{d\gamma^{[x]}}{-\left(\omega_0^2x-\frac{2k_2}{x^3}\right)},
\end{equation}
If we need $\gamma^{[x]}$ in terms of $(x,t)$, we have that
\begin{equation}\label{gamcomp}
 \gamma^{[x]}=\sqrt{-\omega_0^2x^2-2k_2x^{-2}+C}
\end{equation}
including one parameter $C$. For $\gamma^{[y]}$ we obtain equivalent expressions to \eqref{gamcomp},
\begin{equation}
 \gamma^{[y]}=\sqrt{-\omega_0^2y^2-2k_2y^{-2}+K}
\end{equation}


The sections $\gamma^{[x]}$ and $\gamma^{[y]}$ are a solution of the Hamilton--Jacobi equation corresponding with a nonlinear
oscillator with potential $V_a$.
To contemplate the complete solutions, we make use of \eqref{gamcomp} in terms of the parameter $C$. We construct a diffeomorphism
\begin{equation}\label{compsolHJpot1}
 \begin{array}{ccc}
  \Phi:\mathbb{R}^2\times \mathbb{R}\times \mathbb{R}^2& \rightarrow &  T^{*}\mathbb{R}^2\times \mathbb{R}  \\
  \Phi(x,y,t,C,K) &\mapsto &  \Phi(x,y,\gamma^{[x]}(C),\gamma^{[y]}(K),t)
 \end{array}
 \end{equation}
 such that
 {\begin{footnotesize}
 \begin{align}
 \Phi(x,&y,t,C,K)= \nonumber \\ 
 &\left(x,y,\gamma^{[x]}=\sqrt{-\omega_0^2x^2-2k_2x^{-2}+C}, \gamma^{[y]}=\sqrt{-\omega_0^2y^2-2k_2y^{-2}+K},t\right)
 \end{align}
 \end{footnotesize}}

\noindent
{\bf (2)} For the second potential $V_b$, the hamiltonian on $T^{*}Q\times \mathbb{R}$ reads
\begin{equation}\label{secondpotential}
 H=\frac{1}{2}(p_x^2+p_y^2)+\frac{1}{2}\omega^2(4x^2+y^2)+k_2x+\frac{k_3}{y^2}.
\end{equation}
The Reeb vector field \eqref{reebeq} on $T^{*}Q\times \mathbb{R}$ for this example reads
\begin{equation*}
 \mathcal{R}_H=\frac{\partial}{\partial t}+\frac{p_x}{m}\frac{\partial}{\partial x}+\frac{p_y}{m}\frac{\partial}{\partial y}-\left(4\omega_0^2x+k_2\right)\frac{\partial}{\partial p_x}\left(\omega_0^2y-\frac{2k_3}{y^3}\right)\frac{\partial}{\partial p_y}
\end{equation*}
and the projected Reeb vector field on $Q\times \mathbb{R}$ is
\begin{equation*}
  \mathcal{R}_H^{\gamma}=\frac{\partial}{\partial t}+\frac{p_x}{m}\frac{\partial}{\partial x}+\frac{p_y}{m}\frac{\partial}{\partial x}
\end{equation*}
Now,
\begin{align*}
 T\gamma(\mathcal{R}_H^{\gamma})=\frac{\partial}{\partial t}+&\frac{\partial \gamma^{[x]}}{\partial t}\frac{\partial}{\partial p_x}+\frac{\partial \gamma^{[y]}}{\partial t}\frac{\partial}{\partial p_y}+\\
 &+\frac{p_x}{m}\left(\frac{\partial}{\partial x}+\frac{\partial \gamma^{[x]}}{\partial x}\frac{\partial}{\partial p_x}\right)+\frac{p_y}{m}\left(\frac{\partial}{\partial y}+\frac{\partial \gamma^{[y]}}{\partial y}\frac{\partial}{\partial p_y}\right)
\end{align*}
whose difference with respect to $\mathcal{R}_H$ gives us the Hamilton--Jacobi equations for the sections $\gamma^{[x]}$ and $\gamma^{[y]}$. These are
\begin{align}\label{nose1}
 &\gamma_t^{[x]}+\frac{1}{m}\gamma^{[x]}\gamma_x^{[x]}=-\left(4x\omega_0^2+k_2\right),\nonumber\\
  &\gamma_t^{[y]}+\frac{1}{m}\gamma^{[y]}\gamma_y^{[y]}=-\left(\omega_0^2y-\frac{2k_3}{y^3}\right).
\end{align}
The equation for $\gamma^{[y]}$ in \eqref{nose1} has the same solution as \eqref{MPE},  
that is equation \eqref{gamcomp},
\begin{equation}\label{gamcompy}
 \gamma^{[y]}=\sqrt{-\omega_0^2y^2-2k_2y^{-2}+K}
\end{equation}

The equation for $\gamma^{[x]}$ in \eqref{nose1} can be solved by proposing the associated characteristic system
\begin{equation}
 dt=\frac{dx}{\frac{\gamma}{m}}=\frac{d\gamma}{-(4x\omega_0^2+k_2)}.
\end{equation}
If we want to express $\gamma^{[x]}$ in terms of $(x,t)$, we have
\begin{equation}\label{gammaxcomplete1}
 \gamma^{[x]}=m\sqrt{-\left(\frac{4\omega^2}{m}x^2+\frac{2kx}{m}\right)+C}
\end{equation}
with $C$ a parameter of integration.

%

 According to the theory exposed in \eqref{compsolHJ}, if we aim at obtaining complete solutions, we need to construct
 a diffeomorphism based on the one parameter solutions given in \eqref{gammaxcomplete1}.

\begin{equation}\label{compsolHJpot2}
 \begin{array}{ccc}
  \Phi:\mathbb{R}^2\times \mathbb{R}\times \mathbb{R}^2& \rightarrow &  T^{*}\mathbb{R}^2\times \mathbb{R}  \\
  \Phi(x,y,t,C,K) &\mapsto &  \Phi(x,y,\gamma^{[x]}(C),\gamma^{[y]}(K),t)
 \end{array}
 \end{equation}
 such that
 {\begin{footnotesize}
 \begin{align}
 \Phi(x,&y,t,C,K)= \nonumber \\ 
 &\left(x,y,\gamma^{[x]}=m\sqrt{-\left(\frac{4\omega^2}{m}x^2+\frac{2Kx}{m}\right)+C}, \gamma^{[y]}=\sqrt{-\omega_0^2y^2-2k_2y^{-2}+K},t\right)
 \end{align}
 \end{footnotesize}

\section{Hamilton--Jacobi theory on contact manifolds}

We consider the extended phase space $T^{*}Q\times \mathbb{R}$  with canonical projections of the first and second variables
$\rho:T^{*}Q\times \mathbb{R}\rightarrow T^{*}Q$ and $t:T^{*}Q\times \mathbb{R}\rightarrow \mathbb{R}$. The hamiltonian function is
$H:T^{*}Q\times \mathbb{R}\rightarrow \mathbb{R}$. It can be illustrated through the following diagram

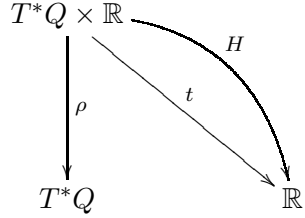
\[
\xymatrix{ T^{*}Q\times \mathbb{R}
\ar[dd]^{\rho} \ar[ddrr]^{t}\ar@/^2pc/[ddrr]^{H}\\
  &  & &\\
T^{*}Q &  & \mathbb{R}}
\]

\bigskip
\noindent
We have local canonical coordinates $\{q^i,p_i,t\}, i=1,\dots,n$. 
The one-form is $\eta=dt-\rho^{*}\theta_Q$, which reads
\begin{equation}\label{contactoneform}
 \eta=dt-\sum_{i=1}^n p_idq^i.
\end{equation}

The pair $(T^{*}Q\times \mathbb{R},\eta)$ is a contact manifold.
The Reeb vector field satisfying
$$\iota_{\mathcal{R}}\eta=1,\quad \iota_{\mathcal{R}}d\eta=0$$
is 
$$\mathcal{R}=\frac{\partial}{\partial t}.$$

Consider the fibration $\pi: T^{*}Q\times \mathbb{R}\rightarrow Q\times \mathbb{R}$.
To have dynamics, we consider the vector field
\begin{equation}
 X_H=\sharp_{\Lambda}(dH)+H\mathcal{R}.
\end{equation}
In coordinates \cite{mejicanos}, it reads
{\begin{footnotesize}
\begin{equation}\label{1hvf}
 X_H=\sum_{i=1}^n\left(p_i\frac{\partial H}{\partial p_i}-H\right)\frac{\partial}{\partial t}-\sum_{i=1}^n\left(p_i\frac{\partial H}{\partial t}+\frac{\partial H}{\partial q^i}\right)\frac{\partial}{\partial p_i}+\sum_{i=1}^n\frac{\partial H}{\partial p_i}\frac{\partial}{\partial q^i}
\end{equation}
\end{footnotesize}}
that is compatible with \eqref{lambdaex} and furthermore it that satisfies the conditions
\begin{equation*}
 \flat{(X_H)}=-(\mathcal{R}(H)+H)\eta+dH.
\end{equation*}
where $\flat$ is the isomorphism defined in \eqref{betamorph} and 
\begin{equation}\label{1exph}
 \eta(X_H)=-H.
\end{equation}
Recall that $(T^{*}Q\times \mathbb{R},\Lambda,\mathcal{R})$ is a Jacobi manifold with $\Lambda$ given in \eqref{lambdaex}.
The proposed contact structure provides us with the {\it dissipation Hamilton equations} \cite{mejicanos}.

\begin{equation}\label{hamileq}
\left\{\begin{aligned}
 {\dot q}^i&=\frac{\partial H}{\partial p_i},\\
 {\dot p}_i&=-\frac{\partial H}{\partial q^i}-p_i\frac{\partial H}{\partial t},\\
{\dot t}&=p_i\frac{\partial H}{\partial p_i}-H.
 \end{aligned}\right.
 \end{equation}
for all $i=1,\dots,n$. 

Consider $\gamma$ a section of $\pi:T^{*}Q\times \mathbb{R} \rightarrow Q\times \mathbb{R}$, i.e., $\pi\circ \gamma=\text{id}_{Q\times \mathbb{R}}$. We can use $\gamma$ to project $X_H$ on $Q\times \mathbb{R}$
just defining a vector field $X_{H}^{\gamma}$ on $Q\times \mathbb{R}$ by
\begin{equation}\label{hjpar}
 X_H^{\gamma}=T_{\pi}\circ X_{H}\circ \gamma.
\end{equation}
The following diagram summarizes the above construction
\[
\xymatrix{ T^{*}Q\times \mathbb{R}
\ar[dd]^{\pi} \ar[rrr]^{X_H}&   & &T(T^{*}Q\times \mathbb{R})\ar[dd]^{T_{\pi}}\\
  &  & &\\
Q\times \mathbb{R} \ar@/^2pc/[uu]^{\gamma}\ar[rrr]^{X^{\gamma}_H}&  & & T(Q\times \mathbb{R})}
\]

Assume that $\gamma(Q\times \mathbb{R})$ is a legendrian submanifold of $(T^{*}Q\times \mathbb{R}, \eta)$, such that $\gamma_t$ is closed.

\begin{theorem}
 The vector fields $X_H$ and $X_H^{\gamma}$ are $\gamma$-related if and only if the following equation is satisfied
\begin{equation}\label{eqtheorem}
[d(H\circ \gamma)]^{V}=-H\dot{\gamma}_q
\end{equation}
where $[\dots]^{V}$ denotes the vertical lift of a one-form and $\dot{\gamma}_q$ is the tangent vector in a point $q$ associated with the curve
\[
\xymatrix{
 \mathbb{R}\ar[rr]\ar@/^2pc/[rrrrrr]^{\gamma_q} & &  Q\times \mathbb{R}\ar[rr] & &  T^{*}Q\times \mathbb{R}\ar[rr]^{\rho} & &  T^{*}Q}
\]

\end{theorem}

\begin{proof}
 The vector fields $X_H$ and $X_H^{\gamma}$ are $\gamma$ related if $T\gamma (X_H^{\gamma})=X_H$. That is,
\begin{equation}\label{transf1}
T\gamma (X_H^{\gamma})=\left(p_i\frac{\partial H}{\partial p_i}-H\right)T\gamma\left(\frac{\partial}{\partial t}\right)+\frac{\partial H}{\partial p_i}T\gamma\left(\frac{\partial}{\partial q^i}\right)=X_H\circ \gamma
\end{equation}
The section $\gamma$ in local coordinates has components $(q^i,\gamma^j(q^i,t),t)$ with $i,j=1,\dots,n$ such that 
\begin{equation}\label{tangent3}
T\gamma\left(\frac{\partial}{\partial t}\right)=\frac{\partial}{\partial t}+\sum_{j=1}^n\frac{\partial \gamma^j}{\partial t}\frac{\partial}{\partial p_j},\quad T\gamma\left(\frac{\partial}{\partial q^i}\right)=\frac{\partial}{\partial q^i}+\sum_{j=1}^n \frac{\partial \gamma^j}{\partial q^i}\frac{\partial}{\partial p_j},\\
\end{equation}
Introducing \eqref{tangent3} in equation \eqref{transf1}, it is straightforward to retrieve condition \eqref{eqtheorem} if 
a further condition on the one-form $\gamma$ is imposed. It is
\begin{equation}
 d\gamma_t=0.
\end{equation}
This means that $\gamma_t$ is closed and fulfills the legendrian submanifold condition.
\end{proof}
Equation \eqref{eqtheorem} is known as a {\it Hamilton--Jacobi equation with respect to a contact structure.}
In local coordinates, 
\begin{equation}
 p_j\frac{\partial H}{\partial t}+\frac{\partial H}{\partial q^j}+\sum_{i=1}^n \left(p_i\frac{\partial H}{\partial p_i}-H\right)\frac{\partial \gamma^j}{\partial t}+\sum_{i=1}^n \frac{\partial H}{\partial p_i}\frac{\partial \gamma^j}{\partial q^i}=0
\end{equation}

%
%

\subsection{Complete solutions}

\begin{definition}
 A {\it complete solution} of the Hamilton--Jacobi equation on a contact manifold $(M,\eta)$ 
 is a diffeomorphism $\Phi:Q\times \mathbb{R}\times \mathbb{R}^n\rightarrow T^{*}Q\times \mathbb{R}$ such that for a set of
 parameters $\lambda\in \mathbb{R}^n, \lambda=(\lambda_1,\dots,\lambda_n)$, the mapping
 
 \begin{equation}
 \begin{array}{ccc}
  \Phi_{\lambda}:Q\times \mathbb{R}& \rightarrow &  T^{*}Q\times \mathbb{R}  \\
  \Phi_{\lambda}(q,t) &\mapsto &  \Phi(q,\gamma(q,t),t)
 \end{array}
 \end{equation}
\noindent
is a solution of the Hamilton--Jacobi equation.

\medskip

We have the following diagram

\[
\xymatrix{ Q\times \mathbb{R}\times \mathbb{R}^n
\ar[dd]^{\alpha} \ar[rrr]^{\Phi}&   & &T^{*}Q\times \mathbb{R}\ar[dd]^{f_i}\ar[lll]^{\Phi^{-1}}\\
  &  & &\\
 \mathbb{R}^n \ar[rrr]^{\pi_i}&  & & \mathbb{R}}
\]
\noindent
where we define functions $f_i$ such that for a point $p\in T^{*}Q\times \mathbb{R}$, it is satisfied
\begin{equation}\label{functions}
 f_i(p)=\pi_i\circ \alpha\circ \Phi^{-1}(p).
\end{equation}
and $\alpha:Q\times \mathbb{R}\times \mathbb{R}^n\rightarrow \mathbb{R}^n$ is the canonical projection.
\end{definition}

\begin{theorem}
 There exist no linearly independent commuting set of first-integrals in involution \eqref{functions} for a complete solution of the Hamilton--Jacobi
 equation on a contact manifold.
\end{theorem}
\begin{proof}
 Consider the bracket
 \begin{equation}
 \{f_i,f_j\}=\Lambda(df_i,df_j)+f_i\mathcal{R}(f_j)-f_j\mathcal{R}(f_i)
 \end{equation}
 On the other hand, recall the definition of the hamiltonian vector field associated with $f_i\in C^{\infty}$ as $X_{f_i}=\sharp_{\Lambda}(df_i)+f_i\mathcal{R}$, then
 \begin{equation}
  X_{f_i}(f_j)=df_j(X_{f_i})=df_j(\sharp_{\Lambda}(df_i)+f_i\mathcal{R})=df_j(\sharp_{\Lambda}(df_i))+df_j(f_i\mathcal{R})
 \end{equation}
 Knowing that $\sharp_{\Lambda}(df_i)\in T\text{Im} \Phi_{\lambda}$, for all $df_i$ the terms $\Lambda(df_i,df_j)$ and $df_j(\sharp_{\Lambda}(df_i))$ vanish.
 Therefore,
 \begin{equation}
  f_i\mathcal{R}(f_j)-f_j\mathcal{R}(f_i)=0, \qquad \forall i,j=1,\dots,n.
 \end{equation}
 From here, we can discuss two possibilities
 \begin{enumerate}
  \item $\mathcal{R}(f_i)=0$
  \item $f_i\mathcal{R}(f_j)-f_j\mathcal{R}(f_i)=0$
 \end{enumerate}

 From 1. we have that $f_i$ are constants. From 2. we have that $\mathcal{R}(f_i/f_j)=0$ that implies that
 $f_i$ and $f_j$ are not linearly independent.
 \noindent
 {\it Remark:} If we calculate the evolution of the functions along the Hamiltonian flow, that is
 \begin{equation}
  X_H(f_i)=df_i(X_H)=df_i(\sharp_{\Lambda}(dH)+H\mathcal{R})=H(\mathcal{R}(f_i))
 \end{equation}
 since $df_i(\sharp_{\Lambda}(dH))=0$ because it is in $(T\text{Im}\Phi)^{o}$. We conclude that $f_i$ are not constants, given
 $X_{H}(f_i)\neq 0$.
 

\end{proof}

\subsection{Example}
Let us consider the hamiltonian on $T^{*}Q\times \mathbb{R}$ with local coordinates $\{q,p,S\}$, 
\begin{equation}
 H=\frac{p^2}{2m}+V(q)+\alpha S
\end{equation}
This is the corresponding hamiltonian of a {\it damped oscillator} \cite{mejicanos} which is retrieved by \eqref{hamileq}.
Taking the hamiltonian vector field as in \eqref{1hvf}, we have
\begin{equation}
 X_H=\left(\frac{p^2}{2m}-V(q)-\alpha S\right)\frac{\partial}{\partial S}-\left(\alpha p+ V'(q)\right)\frac{\partial}{\partial p}+\frac{p}{m}\frac{\partial}{\partial q}
\end{equation}
We choose a legendrian section $\gamma$ with local components $(q,\gamma(q,S),S)$. And $X_H^{\gamma}$ on $Q\times \mathbb{R}$ reads
\begin{equation}
 X_H^{\gamma}=\left(\frac{p^2}{2m}-V(q)-\alpha S\right)\frac{\partial}{\partial S}+\frac{p}{m}\frac{\partial}{\partial q}
\end{equation}
Using \eqref{hjpar}, we need to perfom the computations
\begin{equation}
 T\gamma\left(\frac{\partial}{\partial S}\right)=\frac{\partial}{\partial S}+\frac{\partial \gamma}{\partial S}\frac{\partial}{\partial p},\quad T\gamma\left(\frac{\partial}{\partial q}\right)=\frac{\partial}{\partial q}+\frac{\partial \gamma}{\partial q}\frac{\partial}{\partial p}
\end{equation}
The Hamilton--Jacobi equation of the damped oscillator reads
\begin{equation}\label{puf1}
 \left(\frac{p^2}{2m}-V(q)-\alpha S\right)\frac{\partial \gamma}{\partial S}+\frac{p}{m}\frac{\partial \gamma}{\partial q}+(p\alpha+V'(q))=0
\end{equation}
with $d\gamma_S=0$, that is $\gamma_S=\text{constant}$. Integrating the equations along the section $\gamma$, we have that $p=\gamma$, 
and setting the constant $\gamma_S=1$, then \eqref{puf1} can be rewritten as
\begin{equation}\label{puf2}
 \frac{\partial \gamma}{\partial q}+\frac{1}{2}\gamma+\alpha m+\frac{m}{\gamma}\left(V'(q)-V(q)-\alpha S\right)=0,
\end{equation}
which can be solved as
{\begin{footnotesize}
\begin{equation}\label{gammaimplicit1}
 q=\frac{c_1}{\sqrt{c_1^2-2c_2}}\ln{\left(\frac{\gamma+c_1-\sqrt{c_1^2-2c_2}}{\gamma+c_1+\sqrt{c_1^2-2c_2}}\right)}-\ln{\left(\frac{1}{2}\gamma^2+c_1\gamma+c_2\right)}+C
\end{equation}
\end{footnotesize}}
when $c_1^2>2c_2$, and
{\begin{footnotesize}
\begin{equation}\label{gammaimplicit2}
 q=\frac{2}{\sqrt{2c_2-c_1^2}}\tan^{-1}{\left(\frac{\gamma+c_1}{\sqrt{2c_2-c_1^2}}\right)}-\ln{\left(\frac{1}{2}\gamma^2+c_1\gamma+c_2\right)}+C
\end{equation}
\end{footnotesize}}
when $c_2>\frac{c_1^2}{2}$, with
\begin{equation}
 c_1=\alpha m,\quad c_2=-m^2\alpha S.
\end{equation}
where $V(q)=V'(q)=0$ have conveniently been chosen equal to zero for the possible analytical integration.
\noindent
The solution $\gamma$ of the Hamilton--Jacobi equation on a contact manifold for a damped oscillator is provided by the implicit equations \eqref{gammaimplicit1} and \eqref{gammaimplicit2}.

For a complete solution, we need to construct the diffeomorphism
\begin{equation}
 \begin{array}{ccc}
  \Phi:\mathbb{R}\times \mathbb{R}\times \mathbb{R} &\rightarrow &  T^{*}\mathbb{R}\times \mathbb{R}  \\
  \Phi(q,t,C) &\mapsto &  \Phi(q,\gamma(C),t)
 \end{array}
 \end{equation}
\noindent
with $\gamma(q,t)$ derived from the implicit equations \eqref{gammaimplicit1} and \eqref{gammaimplicit2} for $c_1^2>2c_2$ and $c_2>\frac{c_1^2}{2}$, correspondingly.

\section{Conclusions}
We have developed a two-fold geometric Hamilton--Jacobi theory: for time-dependent hamiltonians through a cosymplectic geometric formalism and for dissipative
hamiltonians through a contact geometry formalism.
 We have derived an explicit, new expression
for the Hamilton--Jacobi equation on a cosymplectic manifold to find solutions of a unidimensional trigonometric system
and two superintegrable potentials in two dimensions, one of them corresponds with the Winternitz--Smorodinsky oscillator.
Furthermore, we have developed a geometric Hamilton--Jacobi theory on a contact manifold for hamiltonians containing a dissipation term.
We have derived an explicit, new expression for the Hamilton--Jacobi
equation on a contact manifold to find solutions of a one-dimensional damped oscillator.

\section*{Acknowledgements}
This work has been partially supported by MINECO MTM 2013-42-870-P and
the ICMAT Severo Ochoa project SEV-2011-0087.

\end{document}